\documentclass[10pt]{article}

\usepackage[T1]{fontenc}										
\usepackage{lmodern}												
\usepackage{microtype}											
\usepackage[LGRgreek,frenchmath]{mathastext}
\usepackage[mathscr]{eucal}                 
\usepackage[sans]{dsfont}										
\usepackage{bbold}													
\usepackage{bbm}														
\usepackage{graphicx}                      	
\usepackage{mdwlist}												
\usepackage{natbib}    											
\setcitestyle{authoryear}
\usepackage[dvipsnames]{xcolor}							
\usepackage[plainpages=false, pdfpagelabels, backref=page]{hyperref} 
	\hypersetup{
		colorlinks   = true,
		citecolor    = RoyalBlue, 
		linkcolor    = RubineRed, 
		urlcolor     = Turquoise
	}
\usepackage[paperwidth=8.5in,paperheight=11.00in,top=1.00in, bottom=1.00in, left=0.75in, right=0.75in]{geometry}
\usepackage{mathtools}                      
\mathtoolsset{showonlyrefs=true}            
\usepackage{amsmath}
\usepackage{amssymb}
\usepackage{amsfonts}
\usepackage{amsthm}                         
\allowdisplaybreaks                         
\newtheoremstyle{plain}
  {}   				
  {}   				
  {\itshape}  
  {}       		
  {\mdseries\scshape} 
  {.}         
  { } 				
  {\thmname{#1}\thmnumber{ #2}\ifx#3\empty\else\ (#3)\fi}
\theoremstyle{plain}
\newtheorem{theorem}{\underline{Theorem}}
       	
\newtheorem{proposition}[theorem]{\underline{Proposition}}

\newtheoremstyle{definition}
  {}   				
  {}   				
  {}  				
  {}      		
  {\mdseries\scshape} 
  {.}         
  { } 				
  {\thmname{#1}\thmnumber{ #2}\ifx#3\empty\else\ (#3)\fi}
\theoremstyle{definition}

\newtheorem{example}[theorem]{\underline{Example}}
\newtheorem{remark}[theorem]{\underline{Remark}}

\usepackage{sectsty}
\allsectionsfont{\mdseries\scshape}

\newcommand{\<}{\langle}
\renewcommand{\>}{\rangle}
\renewcommand{\(}{\left(}
\renewcommand{\)}{\right)}
\renewcommand{\[}{\left[}

\newcommand\Bb{\mathds{B}}

\newcommand\Eb{\mathds{E}}
\newcommand\Fb{\mathds{F}}
\newcommand\Gb{\mathds{G}}
\newcommand\Pb{\mathds{P}}
\newcommand\Qb{\mathds{Q}}
\newcommand\Rb{\mathds{R}}
\newcommand\Ab{\mathds{A}}

\newcommand\Fc{\mathscr{F}}
\newcommand\Gc{\mathscr{G}}

\newcommand\Pc{\mathscr{P}}

\newcommand\Om{\Omega}

\newcommand\gam{\gamma}

\newcommand\lam{\lambda}
\newcommand\del{\delta}
\newcommand\Del{\Delta}

\newcommand\Tb{\overline{T}}

\newcommand\Nt{\widetilde{N}}

\renewcommand\d{\partial}

\newcommand\dd{\mathrm{d}}
\newcommand\ee{\mathrm{e}}

\begin{document}

\title{Interest rate derivatives in a CTMC setting: pricing, replication and Ross recovery}

\author{
Tim Leung
\thanks{Department of Applied Mathematics, University of Washington.  \textbf{e-mail}: \url{timleung@uw.edu}}
\and
Matthew Lorig
\thanks{Department of Applied Mathematics, University of Washington.  \textbf{e-mail}: \url{mlorig@uw.edu}}
}

\date{This version: \today}

\maketitle

\begin{abstract}
We consider a financial market in which the short rate is modeled by a continuous time Markov chain (CTMC) with a finite state space.  In this setting, we show how to price any financial derivative whose payoff is a function of the state of the underlying CTMC at the maturity date.  We also show how to replicate such claims by trading only a money market account and zero-coupon bonds.  Finally, using an extension of Ross' Recovery Theorem due to Qin and Linetsky, we deduce the real-world dynamics of the CTMC.
\end{abstract}

%
%

\section{Introduction}
In order to model the yield curve, banks and researches often assume the short-rate is a Markov diffusion -- typically an affine term structure (ATS) or quadratic term-structure (QTS) model; see, e.g., \cite[Chapter 5]{filipovic2009term} and \cite[Section 2]{ahn2002quadratic} for an overview of such models.  The Markov diffusion framework is appealing because it allows for explicit computation of bond prices, yields, as well as widely-traded interest rate derivatives such as caplets and floorlets.
\\[0.5em]
However, if one looks at time series data of interest rates, they do not appear to have diffusion-like dynamics.  Indeed, we see in Figure \ref{fig:fed-funds} that the federal Funds rate is constant for long periods and then jumps.  As such, to the extent that the federal funds rate is a proxy for the short-rate, it would be better to model the short-rate as a continuous-time Markov chain (CTMC).  The analytic tractability of CTMC short-rate models was established by \cite{elliot-mamon}, who show how to explicitly compute bond prices, yields, and forward rates.  However, they do not consider the problem of replicating interest rate derivatives with liquidly traded assets nor to they comment on the relation between risk-neutral and real-world dynamics of interest rates, both of which are addressed in the present paper.
\\[0.5em]
The rest of this paper proceeds as follows
In Section \ref{sec:model}, we introduce a model for the short-rate driven by a CTMC.
Next, in Section \ref{sec:pricing}, we derive the prices of claims written on the CTMC.  Examples of such claims include zero-coupon bonds, caplets and floorlets.
In Section \ref{sec:replication}, we show how to replicate all claims by dynamically trading a portfolio consisting of bonds with distinct maturity dates and a money market account.
Next, in Section \ref{sec:ross}, we use an extension of Ross' Recovery Theorem \cite{ross}, due to \cite{linetsky-qin}, to deduce the real-world dynamics of the CTMC that drives short rate dynamics. In a related study, \cite{carryu2012} show how to recover real-world transition probabilities in a bounded time-homogeneous diffusion setup with restrictions on the num\'eraire. portfolio. 
Lastly, in Section \ref{sec:examples}, we perform explicit computations in a simple two-state CTMC setting.

%
%

\section{A CTMC short-rate model and assumptions}
\label{sec:model}
Throughout this paper, we work in the same setting as \cite{elliot-mamon}.  Specifically, we fix a time horizon $\Tb < \infty$ and consider a continuous-time financial market, defined on a filtered probability space $(\Om,\Fc,\Fb,\Qb)$ with no arbitrages and no transaction costs.  The probability measure $\Qb$ represents the market's chosen pricing measure taking the \textit{money market account} $M = (M_t)_{t \in [0, \Tb]}$ as num\'eraire.  The filtration $\Fb = (\Fc_t)_{t \in [0,\Tb]}$ represents the history of the market.
\\[0.5em]
We suppose that the money market account $M$ is strictly positive, continuous and non-decreasing.  As such, there exists a non-negative $\Fb$-adapted \textit{short-rate} process $R = (R_t)_{t \in [0,\Tb]}$ such that
\begin{align}
\dd M_t
	&=	R_t M_t \, \dd t ,  &
M_0
	&> 0 . \label{eq:dM}
\end{align}
We will focus on the case in which the dynamics of the short-rate $R$ are described by an irreducible positive recurrent CTMC $J = (J_t)_{t \in [0,\Tb]}$ with state space $S = \{1,2,\ldots,n\}$.  Specifically, we have
\begin{align}
R_t
	&=	r(J_T)  , &
r
	&: S \to [0,\infty) .
\end{align} 
Let $\Gb$ denote the \textit{generator matrix} of $J$
\begin{align}
\Gb
	&=	\( \begin{array}{cccc}
			g_{1,1} & g_{1,2} & \ldots & g_{1,n} \\
			g_{2,1} & g_{2,2} & \ldots & g_{2,n} \\
			\vdots & \vdots & \ddots & \vdots \\
			g_{n,1} & g_{n,2} & \ldots & g_{n,n} 
			\end{array} \) , &
\sum_j g_{i,j} 
	&= 0 \quad \forall \, i, &
g_{i,j}
	&\geq 0 \quad \forall \, i \neq j .
\end{align}
We can express the dynamics of $J$ as a state-dependent L\'evy-type process
\begin{align}
\dd J_t
	&=	\int z N(\dd t, J_{t-},\dd z) , &
\Eb_{t-} N(\dd t, J_{t-},\dd z)
	&=	\nu(J_{t-},\dd z) \dd t , &
\nu(i,\dd z)
	&=	\sum_{j \neq i} g_{i,j} \del_{j-i}(z) \dd z , \label{eq:dJ}
\end{align}
where $\Eb_{t-} \, \cdot \, := \Eb[ \, \cdot \, | \Fc_{t-} ]$ denotes conditional expectation under $\Qb$. 

%
%

\section{Derivative pricing}
\label{sec:pricing}
Consider a financial derivative that pays $\phi(J_T)$ at time $T \leq \Tb$ where $\phi:S \to (-\infty,\infty)$.  Using risk-neutral pricing, the value of this derivative at time $t \leq T$ is $u(t,J_t;T)$ where the function $u(\,\cdot\,,\,\cdot\,;T): [0,T] \times S \to (-\infty,\infty)$ is defined by
\begin{align}
u(t,i;T)
	&:=	\Eb \Big( \ee^{- \int_t^T r(J_s) \dd s} \phi(J_T) \Big| J_t = i \Big) . \label{eq:u-def}
\end{align}
In the following proposition, we provide an explicit expression for $u(t,i;T)$.

\begin{proposition}
\label{prop:u}
 Define $n \times 1$ vectors $U_t^T$ and $\Phi$ and an $n \times n$ diagonal matrix $\Rb$ by
\begin{align}
U_t^T 
	&:= (u(t,i;T))_{i \in S} , &
\Phi 
	&:= (\phi(j))_{j \in S} , &
\Rb
	&:= ( \del_i(j) r(j) )_{i,j \in S} , \label{eq:U-Phi-R}
\end{align}
Then we have
\begin{align}
U_t^T
	&=	\ee^{(T-t) (\Gb - \Rb) } \Phi , \label{eq:UtT}
\end{align}
where $\ee^{(T-t) (\Gb - \Rb) }$ is the matrix exponential of $(T-t) (\Gb - \Rb)$.
\end{proposition}

\begin{proof}
The function $u$ defined in \eqref{eq:u-def} satisfies the Kolmogorov Backward Equation (KBE)
\begin{align}
0
	&=	( \d_t  + \Gc(i) - r(i) ) u(t,i;T) , &
u(T,i;T)
	&=	\phi(i) , \label{eq:u-kbe}
\end{align}
where the operator $\Gc$ is given by
\begin{align}
\Gc(i)
	&=	\int \nu(i,\dd z) \Big( \ee^{z \d} - 1 \Big) 
	=		\sum_{j \neq i} g_{i,j} \Big( \ee^{(j-i) \d} - 1 \Big) . \label{eq:A}
\end{align}
Here, we have used the notation $\ee^{z \d}$ to denote the \textit{shift operator}: $\ee^{z \d} f(j) = f(j+z)$.
Using \eqref{eq:u-kbe} and \eqref{eq:A}, we obtain
\begin{align}
0
	&=	\d_t u(t,i;T) + \sum_{j \neq i} g_{i,j} \Big( u(t,j;T) - u(t,i;T) \Big) - r(i) u(t,i;T) \\
	&=	\d_t u(t,i;T) + \sum_{j} g_{i,j} u(t,j;T) - \sum_j \del_i(j) r(j) u(t,j;T) \\
	&=	\d_t u(t,i;T) + \sum_{j} \Big( g_{i,j} -  \del_i(j) r(j) \Big) u(t,j;T) , \label{eq:u-odes}
\end{align}
where, in the second equality, we have used the fact that $\sum_{j \neq i} g_{i,j} =	- g_{i,i}$.
Using \eqref{eq:U-Phi-R}, we can write the system of coupled ODEs \eqref{eq:u-odes} more compactly as
\begin{align}
0
	&=	( \d_t + \Gb - \Rb ) U_t^T , &
U_T^T
	&=	\Phi ,
\end{align}
from which expression \eqref{eq:UtT} for $U_t^T$ directly follows.
\end{proof}

\begin{remark}
\label{rmk:ui}
For the purposes of computation, it may be convenient to express $u(t,i;T)$ as follows
\begin{align}
u(t,i;T)
	&=	\< E_i , U_t^T \> , &
\< a , b \>
	&:= a^\top b = \sum_{i=1}^n a_i b_i , &
E_i
	&=	(\del_i(j))_{j \in S} .
\end{align}
\end{remark}

\noindent
Let us look at a few examples, which will be important in subsequent sections.

\begin{example}[Zero-coupon bond prices and Yields]
A \textit{$T$-maturity zero-coupon bond} is a derivative that pays $\phi(J_T) = 1$ at time $T \leq \Tb$.
Denote by $B(t,i;T)$ the price of a zero-coupon bond assuming $J_t = i$.  We have
\begin{align}
B(t,i;T)
	&:=	\Eb \Big( \ee^{- \int_t^T r(J_s) \dd s} \Big| J_t = i \Big) .
\end{align}
Defining $n \times 1$ vectors $B_t^T$ and $\mathds{1}$ as follows
\begin{align}
B_t^T 
	&:= (B(t,i;T))_{i \in S} , &
\mathds{1} 
	&:= (1)_{j \in S} ,
\end{align}
we have from Proposition \ref{prop:u} and Remark \ref{rmk:ui} that
\begin{align}
B_t^T
	&=	\ee^{ (T-t) (\Gb - \Rb) } \mathds{1} , &
B(t,i;T)
	&=	\< E_i , B_t^T \> , \label{eq:BtT}
\end{align}
which is equivalent to the \cite[Equation (2)]{elliot-mamon}.  We can also define the \textit{Yield} of a zero-coupon bond as
\begin{align}
Y(t,i;T)
	&:=	\frac{-1}{T-t} \log B(t,i;T) . \label{eq:yield}
\end{align}
\end{example}

\begin{example}[Caplets and Floorlets]
The \textit{simple forward rate from $T$ to $\Tb$} is defined as follows
\begin{align}
F(t,J_t;T,\Tb)
	&:=	\frac{1}{\Tb-T} \Big( \frac{B(t,J_t;T)}{B(t,J_t;\Tb)} - 1 \Big) .
\end{align}
A \textit{forward rate option} with \textit{reset date $T$} and \textit{settlement date $\Tb$} is a derivative that pays $h( F(T,J_T;T,\Tb) )$ at time $\Tb$ for some function $h:\Rb_+ \to \Rb$.  As the payoff to be made at time $\Tb$ is known at time $T$, the value of the forward rate option at time $T$ is $B(T,J_T;\Tb) h(F(T,J_T;T,\Tb))$.  Thus, the value of $v(t,i;T,\Tb)$ of a forward rate option at time $t \leq T$ given $J_t = i$ is
\begin{align}
v(t,i;T,\Tb)
	&=	\Eb \Big( \ee^{- \int_t^T r(J_s) \dd s} \psi(J_T;\Tb) \Big| J_t = i \Big) , &
\psi(J_T;\Tb)
	&:=	B(T,J_T;\Tb) h(F(T,J_T;T,\Tb)) .
\end{align}
Defining $n \times 1$ vectors $V_t^{T,T}$ and $\Psi^{\Tb}$ as follows
\begin{align}
V_t^{T,\Tb} 
	&:= (v(t,i;T,\Tb))_{i \in S} , &
\Psi^{\Tb}
	&:= (\psi(j;\Tb))_{j \in S} ,
\end{align}
we have from Proposition \ref{prop:u} and Remark \ref{rmk:ui} that
\begin{align}
V_t^{T,\Tb}
	&=	\ee^{ (T-t) (\Gb - \Rb) } \Psi^{\Tb} , &
v(t,i;T,\Tb)
	&=	\< E_i , V_t^{T,\Tb} \> .
\end{align}
In the case of a \textit{caplet} or \textit{floorlet} with strike $K$ we have $h(F) = ( F - K )^+$ and $h(F) = (K-F)^+$, respectively.
\end{example}

\begin{example}[Arrow-Debreu securities]
An \textit{Arrow-Debreu} security is a derivative that pays $\phi(J_T) = \del_j(J_T)$ at time $T \leq \Tb$ for some $j \in S$.
Denote by $A(t,i;T;j)$ the price of the $j$th Arrow-Debreu security assuming $J_t = i$.  We have from Proposition \ref{prop:u} and Remark \ref{rmk:ui} that
\begin{align}
A(t,i;T,j)
	&:=	\Eb \Big( \ee^{- \int_t^T r(J_s) \dd s} \del_j(J_T) \Big| J_t = i \Big) 
	=		\< E_i , \ee^{(T-t) (\Gb - \Rb) } E_j \> . \label{eq:ptiTj}
\end{align}
Defining the $n \times n$ matrix $\Ab_t^T := (A(t,i;T,j))_{i,j \in S}$, we have
\begin{align}
\Ab_t^T
	&=	\ee^{(T-t) (\Gb - \Rb) } . \label{eq:PtT}
\end{align}
\end{example}


%
%

\section{Replication}
\label{sec:replication}
In this section, we will show how to replicate an Arrow-Debreu security with a payoff $\del_k(J_T)$ at time $T$ by trading a portfolio of assets consisting of the money market account $M$ and $(n-1)$ zero-coupon bonds with maturity dates $T_1, T_2, \ldots, T_{n-1}$ where $T_i \in (T,\Tb)$ and $T_i \neq T_j$.  
Once we establish how to replicate an Arrow-Debreu security with payoff $\del_k(J_T)$, we can replicate a claim with general payoff $\phi(J_T)$ by noting that $\phi(J_T)$ can be written as a linear combination of Arrow-Debreu payoffs:
$\phi(J_T) = \sum_k \phi(k) \del_k(J_T)$.

\begin{proposition}
\label{prop:replication}
Denote by $X = (X_t)_{t \in [0,T]}$ the value of a self-financing portfolio with dynamics of the form
\begin{align}
\dd X_t
	&=	\sum_{i=1}^{n-1} \Del_{t-}^{(i)} \dd B_t(t,J_t;T_i) + \Big( X_t - \sum_{i=1}^{n-1} \Del_{t-}^{(i)} B_t(t,J_t;T_i) \Big) r(J_t) \dd t , \label{eq:dX}
\end{align}
where $\Del^{(i)} = (\Del_t^{(i)})_{t \in [0,T]}$ denotes the number of $T_i$-maturity bonds in the portfolio. 
Define an $(n-1) \times (n-1)$ matrix
\begin{align}
\del \Bb(t,J_{t-})
	&:= \Big( B(t,j,T_i) - B(t,J_{t-},T_i) \Big)_{i,j} , &
i
	&=	1, 2, \ldots, n-1 , &
j
	&=	1, \ldots, J_{t-}-1, J_{t-} + 1, \ldots n , 
\end{align}
and $(n-1) \times 1$ vectors
\begin{align}
D_{t-}
	&=	(\Del_{t-}^{(i)})_i , &
i
	&=	1, 2, \ldots, n-1 , \\
\del A(t,J_{t-};T,k)
	&:=	( A(t,j;T,k) - A(t,J_{t-};T,k) )_j , &
j
	&:=	1, \ldots, J_{t-}-1, J_{t-} + 1, \ldots n .
\end{align}
Suppose $X_0 = A(0,J_0;T,k)$ and $D_t \equiv D(t,J_t;T,k)$ 
satisfies
\begin{align}
\del \Bb(t,J_t) D(t,J_t;T,k)
	&=	\del A(t,J_t;T,k) . \label{eq:matrix-form}
\end{align}
Then $X_t = A(t,J_t;T,k)$ for all $t \in [0,T]$.
\end{proposition}

\begin{proof}
We will show that, if $X_0 = A(0,J_0;T,k)$ and $D_{t-}$ 
satisfies \eqref{eq:matrix-form}, then $\dd (X_t/M_t) = \dd (A(t,J_t;T,k)/M_t)$ from which it follows that $X_t = A(t,J_t;T,k)$ for all $t \in [0,T]$.  We begin by computing the dynamics of $B(t,J_t;T_i)$ and $A(t,J_t;T,k)$.  
Defining the \textit{compensated} state-dependent Poisson random measure
\begin{align}
\Nt(\dd t, i, \dd z)
	&:= N( \dd t , i, \dd z) - \nu(i, \dd z) \dd t ,
\end{align}
and using \eqref{eq:dJ} as well as It\^o's rule for L\'evy-It\^o processes, we obtain
\begin{align}
\dd B(t,J_t;T_i)
	&=	( \ldots ) \dd t + \int \Big( B(t,J_{t-}+z,T_i) - B(t,J_{t-},T_i) \Big) \Nt(\dd t, J_{t-}, \dd z) ,  \label{eq:dB} \\
\dd A(t,J_t;T,k)
	&=	( \ldots ) \dd t + \int \Big( A(t,J_{t-}+z,T,k) - A(t,J_{t-},T,k) \Big) \Nt(\dd t, J_{t-}, \dd z) ,		\label{eq:dp} 
\end{align}
where the $\dd t$ terms will not be important.  Next, using \eqref{eq:dM}, \eqref{eq:dB} an \eqref{eq:dp}, as well as the product rule for L\'evy-It\^o processes, we find
\begin{align}
\dd \Big( \frac{B(t,J_t;T_i)}{M_t} \Big)
	&=	\frac{1}{M_t} \int \Big( B(t,J_{t-}+z,T_i) - B(t,J_{t-},T_i) \Big) \Nt(\dd t, J_{t-}, \dd z)  \\
	&=	\frac{1}{M_t} \dd B(t,J_t;T_i) + (\ldots) \dd t , \label{eq:dB-over-M} \\
\dd \Big( \frac{A(t,J_t;T,k)}{M_t} \Big)
	&=	\frac{1}{M_t} \int \Big( A(t,J_{t-}+z,T,k) - A(t,J_{t-},T,k) \Big) \Nt(\dd t, J_{t-}, \dd z) , \label{eq:dp-over-M}
\end{align}
Similarly, using \eqref{eq:dM} and \eqref{eq:dX} as well as the product rule for L\'evy-It\^o processes, we obtain
\begin{align}
&\dd \Big( \frac{X_t}{M_t} \Big)
	=	\frac{1}{M_t} \dd X_t + X_t \dd \Big( \frac{1}{M_t} \Big) + \dd \Big[ X, \frac{1}{M} \Big]_t \\
	&=	\frac{1}{M_t} \dd X_t - r(J_t) \Big( \frac{X_t}{M_t} \Big) \dd t &
	&		\text{(as $\dd \Big[ X, \frac{1}{M} \Big]_t = 0$)} \\
	&=	\frac{1}{M_t} \sum_{i=1}^{n-1} \Del_{t-}^{(i)} \dd B_t(t,J_t;T_i) 
			+ \frac{1}{M_t} \Big( X_t - \sum_{i=1}^{n-1} \Del_{t-}^{(i)} B_t(t,J_t;T_i) \Big) r(J_t) \dd t
			- r(J_t) \Big( \frac{X_t}{M_t} \Big) \dd t &
	&		\text{(by eq. \eqref{eq:dX})} \\
	&=	\frac{1}{M_t} \sum_{i=1}^{n-1} \Del_{t-}^{(i)} \dd B_t(t,J_t;T_i) 
			- \frac{1}{M_t} \sum_{i=1}^{n-1} \Del_{t-}^{(i)} B_t(t,J_t;T_i) r(J_t) \dd t &
	&		\text{(canceling $\frac{r(J_t) X_t }{ M_t } \dd t$)} \\
	&=	\sum_{i=1}^{n-1} \Del_{t-}^{(i)} \Big( \dd \Big( \frac{B_t(t,J_t;T_i) }{M_t} \Big) - \frac{1}{M_t} ( \ldots ) \dd t \Big)
			- \frac{1}{M_t} \sum_{i=1}^{n-1} \Del_{t-}^{(i)} B_t(t,J_t;T_i) r(J_t) \dd t &
	&		\text{(by eq. \eqref{eq:dB-over-M})} \\
	&=	\frac{1}{M_t} \sum_{i=1}^{n-1} \Del_{t-}^{(i)}  \int \Big( B(t,J_{t-}+z,T_i) - B(t,J_{t-},T_i) \Big) \Nt(\dd t, J_{t-}, \dd z) , \label{eq:dX-over-M}
\end{align}
where, in the last equality, we have used \eqref{eq:dB-over-M} and the fact that the $\dd t$ terms must cancel, as $X/M$ is a $(\Qb,\Fb)$ martingale.
Comparing \eqref{eq:dp-over-M} an \eqref{eq:dX-over-M}, we see that in order for $\dd (X_t/M_t) = \dd (A(t,J_t;T,k)/M_t)$, we must have
\begin{align}
&\sum_{i=1}^{n-1} \Del_{t-}^{(i)}  \int \Big( B(t,J_{t-}+z,T_i) - B(t,J_{t-},T_i) \Big) \Nt(\dd t, J_{t-}, \dd z) \\
	&= \int \Big( A(t,J_{t-}+z,T,k) - A(t,J_{t-},T,k) \Big) \Nt(\dd t, J_{t-}, \dd z) .
\end{align}
Noting from \eqref{eq:dJ} that $J$ can only jump to a state $j \neq J_{t-}$, we can write the above equality as
\begin{align}
\sum_{i=1}^{n-1} \Del_{t-}^{(i)} \Big( B(t,j,T_i) - B(t,J_{t-},T_i) \Big)
	&=	A(t,j;T,k) - A(t,J_{t-};T,k) , &
	&\forall \, j \neq J_{t-} .
\end{align}
which, in matrix form, is \eqref{eq:matrix-form}.
\end{proof}

\begin{remark}[Completeness of the market]
Observe that, if zero-coupon bonds trade at $(n-1)$ distinct maturities, then the market described in Section \ref{sec:model} is complete (at least, up until the time of the shortest maturity).
\end{remark}

\begin{remark}[Replication with fewer than $(n-1)$ bonds]
If, for all $i \in S$, the CTMC $J$ can only jump to $m-1 \leq n-1$ states, then replication of any European-style claim can be achieved by trading $m-1$ bonds and the money market account. For example if, for any jump time $\tau$, we have $J_\tau - J_{\tau-} \in \{-1,1\}$, then any European-style claims can be replicated by trading two bonds on the money market account.
\end{remark}

%
%

\section{Ross recovery}
\label{sec:ross}
In this section, we use an extension of the Ross Recovery Theorem \cite{ross}, due to \cite{linetsky-qin}, in order to determine the dynamics of $J$ under the physical (i.e., real-world) probability measure $\Pb$.  While the dynamics of $J$ under $\Pb$ are not important for the purposes of pricing and replicating financial derivatives, real-world dynamics are important for designed optimal investment strategies.

\begin{proposition}
\label{thm:ross}
Suppose that, for some constant $\rho < 0$, we have
\begin{align}
( \Gb - \Rb ) \pi
	&=	\rho \pi , 
\end{align}
where $\pi(i) > 0$ for all $i \in S$.  Then the generator matrix of $J$ under the real-world probability measure $\Pb$ is given by
\begin{align}
\Gb^\pi
	&=	\( \begin{array}{cccc}
			g_{1,1}^\pi & g_{1,2}^\pi & \ldots & g_{1,n}^\pi \\
			g_{2,1}^\pi & g_{2,2}^\pi & \ldots & g_{2,n}^\pi \\
			\vdots & \vdots & \ddots & \vdots \\
			g_{n,1}^\pi & g_{n,2}^\pi & \ldots & g_{n,n}^\pi 
			\end{array} \) , &
g_{i,j}^\pi
	&= \frac{\pi(j)}{\pi(i)}g_{i,j} , \qquad i \neq j , &
g_{i,i}
	&=	- \sum_{j \neq i} g_{i,j}^\pi . \label{eq:G-pi}
\end{align}
\end{proposition}

\begin{proof}
Defining
\begin{align}
\Pi_t^T 
	&:= \ee^{\rho (T - t)} \pi , &
\Pi(t,i;T)
	&=	\< E_i , \Pi_t^T \> = \ee^{\rho (T-t)} \pi(i) ,
\end{align}
it is easy to see that $\Pi_t^T$ satisfies the pricing equation
\begin{align}
0
	&=	( \d_t + \Gb - \Rb ) \Pi_t^T , &
\Pi_T^T	
	&=	\pi  .
\end{align}
Thus, we have that $\Pi(t,J_t;T)$ is the value at time $t \leq T$ of a derivative that pays $\pi(J_T)$ at time $T$.  It follows that $\Pi(t,J_t;T)/M_t$ is a $(\Qb,\Fb)$-martingale, and we can define a new probability measure $\Pb^\pi$, whose relation to $\Qb$ is characterized by the following Radon-Nikodym derivative process
\begin{align}
\frac{\dd \Pb^\pi}{\dd \Qb} \equiv Z_T
	&:=	\frac{\Pi(T,J_T;T)/M_T}{\Pi(0,J_0;T)/M_0 } 
	=		\ee^{ - \int_0^T r(J_s) \dd s - \rho T} \Big( \frac{\pi(J_T)}{\pi(J_0)} \Big) , &
Z_t
	&=	\Eb_t Z_T .
\end{align}
Thus, denoting by $\Eb^\pi$ expectation under $\Pb^\pi$, the value $u(t,J_t;T)$ of a financial derivative that pays $\phi(J_T)$ at time $T$ satisfies
\begin{align}
\frac{u(t,J_t;T)}{M_t}
	&=	\Eb \Big( \frac{\phi(J_T)}{M_T} \Big| \Fc_t \Big)  
	 =	\Eb^\pi \Big( \frac{Z_t}{Z_T} \frac{\phi(J_T)}{M_T} \Big| \Fc_t \Big)  
	 =	\Eb^\pi \Big( \frac{M_T \Pi(t,J_t;T)}{M_t \Pi(T,J_T;T)} \frac{\phi(J_T)}{M_T} \Big| \Fc_t \Big) .
\end{align}
Canceling common factors of $M_t$ and $M_T$ and using the Markov property of $J$, we find
\begin{align}
u(t,J_t;T)
	&=	\Pi(t,J_t;T) \Eb^\pi \Big( \frac{ \phi(J_T) }{\Pi(T,J_T;T)}  \Big| \Fc_t \Big)
	=		\pi(J_t) \ee^{\rho (T-t)} \Eb^\pi \Big( \frac{ \phi(J_T) }{\pi(J_T)} \Big| J_t  \Big) .
\end{align}
Thus, we have a \textit{transition independent pricing kernel}
\begin{align}
\Pc_{T-t} \phi(i)
	&:=	\pi(i) \ee^{\rho (T-t)} \Eb^\pi \Big( \frac{ \phi(J_T) }{\pi(J_T)} \Big| J_t = i \Big) . \label{eq:TIPK}
\end{align}
Under the Assumptions of Section \ref{sec:model}, when one has a pricing kernel of the form \eqref{eq:TIPK}, we have by \cite[Theorem 3.3]{linetsky-qin} that $\Pb^\pi$ \textit{is} the real-world probability measure: $\Pb^\pi = \Pb$.  
\\[0.5em]
With the aim of deducing the dynamics of $J$ under $\Pb^\pi \equiv \Pb$, we compute
\begin{align}
\dd \Big( \frac{\Pi(t,J_t;T)}{M_t} \Big)
	&=	\frac{1}{M_t} \int \Big( \Pi(t,J_{t-}+z;T) - \Pi(t,J_{t-};T) \Big) \Nt(\dd t, J_{t-}, \dd z) \\
	&=	\frac{\Pi(t,J_{t-};T)}{M_{t-}} \int \Big( \frac{ \Pi(t,J_{t-}+z;T) }{ \Pi(t,J_{t-};T) } - 1 \Big) \Nt(\dd t, J_{t-}, \dd z) \\
	&=	\frac{\Pi(t,J_{t-};T)}{M_{t-}} \int \Big( \ee^{\eta(J_{t-},z)} - 1 \Big) \Nt(\dd t, J_{t-}, \dd z) , \\
\eta(J_{t-},z)
	&:=	\log \Big( \frac{ \Pi(t,J_{t-}+z;T) }{ \Pi(t,J_{t-};T) } \Big) 
	=		\log \Big( \frac{ \pi(J_{t-}+z) }{ \pi(J_{t-}) } \Big) ,
\end{align}
from which it follows that
\begin{align}
\frac{\dd \Pb^\pi}{\dd \Qb}
	&=	\frac{\Pi(T,J_T;T)/M_T}{\Pi(0,J_0;T)/M_0 }
	=		\exp \Big( - \int_0^T \int \Big( \ee^{\eta(J_{t-},z)} - 1 -  \eta(J_{t-},z)  \Big) \nu(J_{t-}, \dd z) \dd t
				+ \int_0^T \int \eta(J_{t-},z) \Nt(\dd t, J_{t-}, \dd z) \Big) .
\end{align}
Hence, by Girsanov's Theorem for L\'evy-It\^o processes (see, e.g., \cite[Theorem 1.35]{oksendal-sulem}), we have
\begin{align}
\Eb_{t-}^\pi N(\dd t,J_{t-}, \dd z)
	&=	\nu^\pi(J_{t-}, \dd z) \dd t , &
\nu^\pi(i, \dd z)
	&:=	\ee^{\eta(i,z)} \nu(i, \dd z)
	=		\sum_{j \neq i} \frac{\pi(j)}{\pi(i)} g_{i,j} \del_{j-i}(z) \dd z . \label{eq:nu-pi}
\end{align}
From \eqref{eq:nu-pi}, we see that the generator matrix of $J$ under $\Pb^\pi$ is given by \eqref{eq:G-pi}.
\end{proof}

%
%

\section{Example: a two-state CTMC}
\label{sec:examples}
Throughout this section, we suppose that the generator matrix $\Gb$ of $J$ under $\Qb$ and the matrix $\Rb$ of interest rates are given by
\begin{align}
\Gb
	&=	\( \begin{array}{cc}
			- \lam & \lam \\
			\lam & - \lam
			\end{array} \) , &
\Rb
	&=	\( \begin{array}{cc}
			0 & 0 \\
			0 & r
			\end{array} \) ,
\end{align}
where $\lam, r > 0$.  The eigenvalues and corresponding ($L^2$ normalized right) eigenvectors of $\Gb - \Rb$, denoted $\rho_\pm$ and $\pi_\pm$, respectively, are given by
\begin{align}
\rho_\pm
	&=	(-2 \lam - r \pm \gam) / 2 , &
\pi_\pm
	&= \frac{1}{\sqrt{(r \pm \gam)^2 + (2 \lam)^2}} \( \begin{array}{c}
			r \pm \gam \\
			2 \lam 
			\end{array} \) , &
\gam
	&:=	\sqrt{ (2 \lam)^2 + r^2 } .
\end{align}
Noting that $(\Gb - \Rb)$ is symmetric, and thus its eigenvectors are orthogonal, we have from \eqref{eq:PtT} that matrix of Arrow-Debreu security prices is
\begin{align}
\Ab_t^T
	&=	\ee^{(T-t)(\Gb-\Rb)}
	 =	\( \begin{array}{cc}
			\pi_+ &
			\pi_-
			\end{array} \)
			\( \begin{array}{cc}
			\ee^{(T-t) \rho_+} &  0 \\
			0 & \ee^{(T-t) \rho_-}
			\end{array} \)
			\( \begin{array}{c}
			\pi_+^\top \\ \pi_-^\top 
			\end{array} \) \\ 
	&=	\frac{1}{2 \gam} \ee^{-\frac{1}{2} (T-t) (\gam +2 \lam +r)} 
		\( \begin{array}{cc}
		(\gam -r) +(\gam +r) \ee^{\gam  (T-t)}  &
		2 \lam  (\ee^{\gam  (T-t)}-1) \\
		2 \lam  (\ee^{\gam  (T-t)}-1) &
		(\gam+r) +(\gam -r) \ee^{\gam  (T-t)}
	\end{array} \) .
\end{align}
Next, we have from \eqref{eq:BtT} that vector of zero-coupon bond prices is
\begin{align}
B_t^T
	&=	\ee^{(T-t)(\Gb-\Rb)} \mathds{1}
	=		\( \begin{array}{cc}
			\pi_+ &
			\pi_-
			\end{array} \)
			\( \begin{array}{cc}
			\ee^{(T-t) \rho_+} &  0 \\
			0 & \ee^{(T-t) \rho_-}
			\end{array} \)
			\( \begin{array}{c}
			\pi_+ ^\top\\ \pi_-^\top
			\end{array} \) 
			\( \begin{array}{c}
			1 \\ 1
			\end{array} \) \\
	&=	\frac{1}{2 \gam} \ee^{-\frac{1}{2} (T-t) (\gam +2 \lam +r)} 
			\( \begin{array}{c}
			( \gam -2 \lam -r ) + \ee^{\gam  (T-t)} (\gam +2 \lam +r)  \\
			( \gam -2 \lam +r ) + \ee^{\gam  (T-t)} (\gam +2 \lam -r)
			\end{array} \) .
\end{align}
From the above expression as well as the definition of the yield \eqref{eq:yield} we have
\begin{align}
\lim_{T \to \infty} Y(t,i;T) \equiv Y(t,i;\infty)
	&=	( r + 2 \lam - \gam)/2 .
\end{align}
In Figure \ref{fig:yield}, we plot the yield $Y(t,i;T)$ as a function of $T$ for $i \in \{1,2\}$ with $t=0$ fixed.  For comparison, we also plot limiting yield $Y(t,i,\infty)$.
\\[0.5em]
As the state space of $J$ is $S=\{1,2\}$, we need only one bond and the money market account to replicate an Arrow-Debreu security.
Thus, dynamics of the replicating portfolio $X$ are of the form
\begin{align}
\dd X_t
	&=	\Del_{t-}^{(1)} \dd B(t,J_t;T_1) + \Big( X_t - \Del_t^{(1)} B(t,J_t;T_1) \Big) r(J_t) \dd t .
\end{align}
Using \eqref{eq:matrix-form}, the number of bonds $\Del_t^{(1)}$ one should hold in the portfolio in order to replicate the $k$th Arrow-Debreu security with maturity $T$ is
\begin{align}
\Del_{t}^{(1)}
	&=	\left\{ \begin{aligned}
			&\frac{ A(t,2;T,k) - A(t,1;T,k) }{ B(t,2;T_1) - B(t,1;T_1) } , &\text{if $J_{t-} = 1$} \\
			&\frac{ A(t,1;T,k) - A(t,2;T,k) }{ B(t,1;T_1) - B(t,2;T_1) } , &\text{if $J_{t-} = 2$}
			\end{aligned} \right\}
	=	\frac{ A(t,2;T,k) - A(t,1;T,k) }{ B(t,2;T_1) - B(t,1;T_1) } ,
\end{align}
which, in this example, happens to be independent of the state of $J$.
\\[0.5em]
In order to find the dynamics of $J$ under the real-world probability measure $\Pb^\pi$ we note that
\begin{align}
(\Gb - \Rb) \pi_+
	&=	\rho_+ \pi_+ , &
\pi_+(i)
	&>	0 , & 
i
	&\in \{1,2\} .
\end{align}
It follows from Proposition \ref{thm:ross} that the generator matrix of $J$ under the real-world measure $\Pb^\pi$ is
\begin{align}
\Gb^\pi
	&=	\( \begin{array}{cc}
			- \frac{\pi_+(2)}{\pi_+(1)} \lam & \frac{\pi_+(2)}{\pi_+(1)} \lam \\
			\frac{\pi_+(1)}{\pi_+(2)} \lam		& - \frac{\pi_+(1)}{\pi_+(2)} \lam
			\end{array} \) 
	 =	\( \begin{array}{cc}
			- \Big( \frac{2 \lam^2}{\gam + r} \Big) & \frac{2 \lam^2}{\gam + r} \\
			\frac{\gam + r }{2}		& - \Big( \frac{\gam + r }{2} \Big) 
			\end{array} \) .
\end{align}

%
%

\bibliography{references}

\clearpage

\begin{figure}
\includegraphics[width=1.0\textwidth]{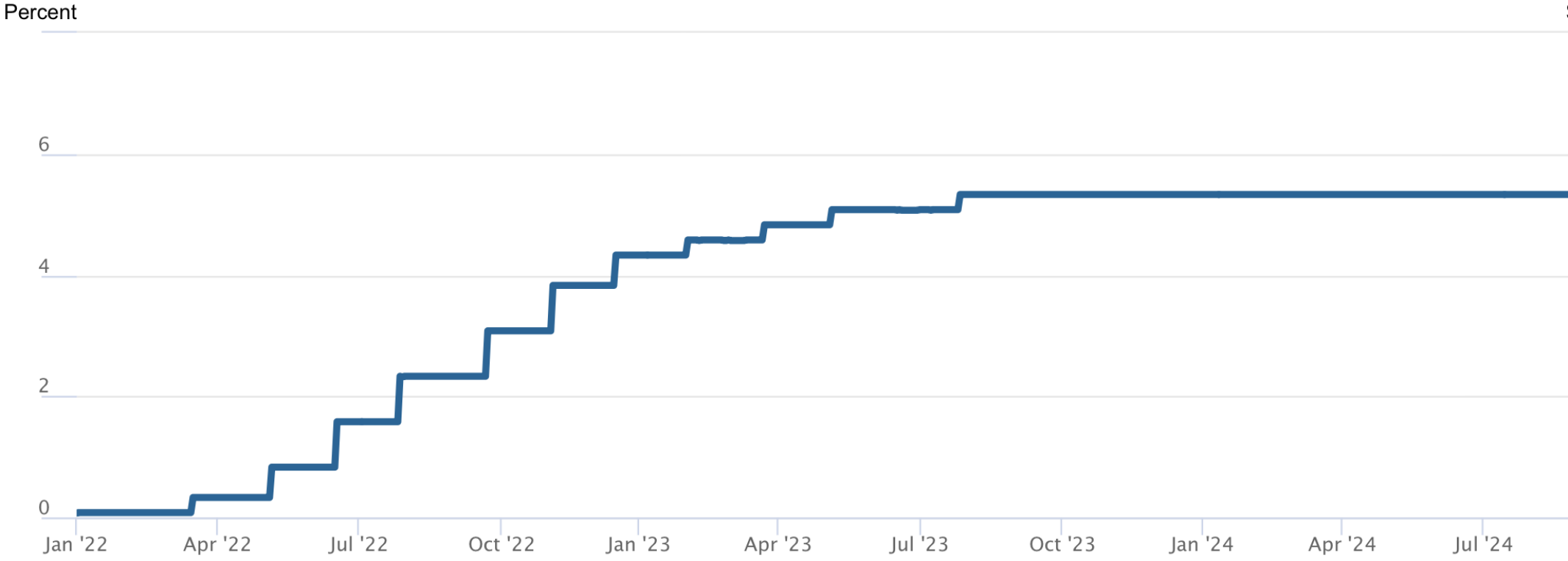}
\caption{Effective federal funds rate (EFFR) from January 1, 2022 to September 1, 2024.  Source: \url{https://www.newyorkfed.org/markets/reference-rates/effr}.}
\label{fig:fed-funds}
\end{figure}

\begin{figure}
\includegraphics[width=1.0\textwidth]{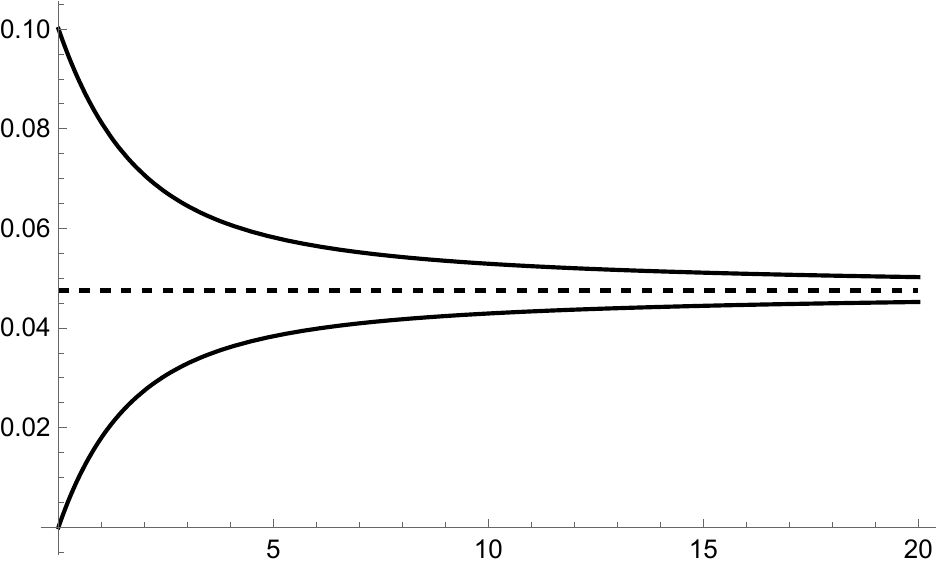}
\caption{For the model considered in Section \ref{sec:examples}, we plot the yield $Y(t,i;T)$ as a function of $T$ with $t=0$ fixed.  The solid curves below and above the dashed line correspond to $i = 1$ and $i=2$, respectively.  The dashed line is the yield in the limit as maturity tends to infinity: $\lim_{T \to \infty} Y(t,i;T)$.  Parameters used in this plot are $\lam=1/2$ and $r = 0.1$.}
\label{fig:yield}
\end{figure}

\end{document}